\documentclass[aps,prl,twocolumn]{revtex4-2}
\usepackage{cancel}
\usepackage{blindtext}
\usepackage[T1]{fontenc}
\usepackage[latin9]{inputenc}
\usepackage{amsthm}
\usepackage{amsmath}
\usepackage{amssymb}
\usepackage{color}
\usepackage{hyperref}
\usepackage{calc}
\usepackage{graphicx}
\usepackage{epstopdf}
\usepackage{mathtools}
\newtheorem{theorem}{Theorem}
\newtheorem{definition}{Definition}
\newtheorem{Corollary}{Corollary}

\newtheorem{fact}{Fact}
\newtheorem{lemma}{Lemma}
\newtheorem{remark}{Remark}

\newcommand{\bk}[2]{\ensuremath{\langle #1 | #2 \rangle}}

\begin{document}
\title{How to check universality of quantum gates?}
\author{Adam Sawicki${}^{1}$, Lorenzo Mattioli${}^{1}$, Zolt\'an Zimbor\'as${}^{2,3}$}
%\email[E-mail: ]{a.sawicki@cft.edu.pl, mattioli@cft.edu.pl}
\affiliation{$^1$Center for Theoretical Physics PAS, Al. Lotnik\'ow 32/46, 02-668, Warszawa, Poland}
\affiliation{$^2$Wigner Research Centre for Physics,H-1525, P.O.Box 49, Budapest, Hungary} \affiliation{$^3$ BME-MTA Lend\"ulet Quantum Information Theory Research Group, Budapest, Hungary}
%\affiliation{$^2$School of Mathematics, University of Bristol, University Walk, Bristol BS8 1TW, UK}
\begin{abstract}
We provide two simple universality criteria. Our first criterion states that $\mathcal{S}\subset G_d:=U(d)$ is universal if and only if $\mathcal{S}$ forms a $\delta$-approximate $t(d)$-design, where $t(2)=6$ and $t(d)=4$ for $d\geq3$. Our second universality criterion says that  $\mathcal{S}\subset G_d$ is universal if and only if the centralizer of $\mathcal{S}^{t(d),t(d)}=\{U^{\otimes t(d)}\otimes \bar{U}^{\otimes t(d)}|U\in \mathcal{S}\}$ is equal to the centralizer of $G_d^{t(d),t(d)}=\{U^{\otimes t(d)}\otimes \bar{U}^{\otimes t(d)}|U\in G_d\}$, where $t(2)=3$, and $t(d)=2$ for $d\geq 3$. The equality of the centralizers can be verified by comparing their dimensions.
\end{abstract}
\maketitle

Universal and efficient quantum gates play a central role in quantum computing \cite{Barenco,Dojcz,Loyd,Reck,exp1,exp2,harrow,Sawicki1,Sawicki3, BA,selinger,klucz, bocharov, sarnak}. It is well known that in order to construct a universal set of gates for many qudits it is enough to take universal set for one qudit and extend it by a two-qudit entangling gate \cite{Brylinski,NC00} (see \cite{Oszmaniec1} for fermionic systems). On the other hand, it is a great challenge to find a time efficient procedure that enables deciding if a given set of gates  $\mathcal{S}\subset G_d:=U(d)$ is universal. One can consider an analogous problem for Hamiltonians, i.e. given a set of anti-Hermitian traceless matrices, $\mathcal{X}\subset \mathfrak{su}(d)$, we want to check if they generate whole algebra of the special unitary group, $\mathfrak{su}(d)$. The answer turns out to be relatively simple and can be phrased in terms of the centralizer of the tensor squares of the elements belonging to $\mathcal{X}$ \cite{daniel, zeier1, zeier} (cf. \cite{Laura,Claudio,Dallesandro}). The approach through tensor powers can be extended from Hamiltonians to quantum gates. To this end one considers the centralizers of $\mathcal{S}^{t,t}=\{U^{\otimes t}\otimes \bar{U}^{\otimes t}|U\in \mathcal{S}\}$ and $G_d^{t,t}=\{U^{\otimes t}\otimes \bar{U}^{\otimes t}|U\in G_d\}$ for any $t$. In \cite{Adam,Sawicki2} it was shown that the equality of the centralizers of $\mathcal{S}^{1,1}$ and $G_d^{1,1}$ implies that $\mathcal{S}$ is universal provided $\mathcal{S}$ generates an infinite subgroup of $G_d$. We will call the equality of the centralizers of $\mathcal{S}^{1,1}$ and $G_d^{1,1}$ the necessary universality condition. For $\mathcal{S}$ satisfying the necessary universality condition, in order to verify if the group generated by $\mathcal{S}$ is infinite it is enough to check if it forms an $\frac{1}{2\sqrt{2}}$-net, in the Hilbert-Schmidt distance \cite{Adam,Sawicki2,Mattioli-Sawicki,Freedman}. Following this idea the authors of \cite{Nets} showed that $\mathcal{S}$ is universal iff 1) $\mathcal{S}\subset G_d$ satisfies the necessary universality condition, 2) gates belonging to $\mathcal{S}$ form a $\delta$-approximate $t(d)$-design, with $\delta<1$  and $t(d)=O(d^{5/2})$. 

In this paper, we show that actually the required $t(d)$ does not grow with $d$. Utilizing the recent results regarding the so-called unitary $t$-groups \cite{Bannai20}, we formulate our first universality criterion (Theorem~\ref{main}). It states that $\mathcal{S}\subset G_d$ is universal if and only if $\mathcal{S}$ forms a $\delta$-approximate $t(d)$-design, where $t(2)=6$ and $t(d)=4$ for $d\geq3$. Using further properties of $t$-designs, we formulate our second universality criterion (Theorem~\ref{main3}) which says that  $\mathcal{S}\subset G_d$ is universal if and only if the centralizer of $\mathcal{S}^{t(d),t(d)}$ is equal to the centralizer $G_d^{t(d),t(d)}$, where $t(2)=3$, and $t(d)=2$ for $d\geq 3$. The equality of the centralizers can be verified by checking their dimensions. We calculate the dimensions of the centralizers of  $G_d^{t(d),t(d)}$ using representation theory methods, for all $t(d)$ listed above and use them in our final criterion for universality, i.e. Theorem \ref{main4}. Thus the problem reduces to finding the dimension of the centralizer of $\mathcal{S}^{t(d),t(d)}$ which is a standard linear algebra task. 

Let  $\{\mathcal{S},\nu_{\mathcal{S}}\}$ be an ensemble of quantum gates, where $\mathcal{S}$ is a finite subset of $G_d:=U(d)$ and $\nu_\mathcal{S}$ is the uniform measure on $\mathcal{S}$ and throughout the paper we assume that $\mathcal{S}$ contains identity. Such ensemble is called 
 $\delta(t,\nu_\mathcal{S})$-approximate $t$-design if and only if 
\begin{equation}\label{eq:DEFexpand}
\delta(t,\nu_{\mathcal{S}})\coloneqq\left\|T_{\nu_\mathcal{S},t}-T_{\mu,t} \right\|_\infty < 1\  ,
\end{equation}
where for any measure $\nu$ (in particular for the Haar measure $\mu$) we define a \emph{moment operator}
 \begin{equation}\label{eq:momentOP}
T_{\nu,t}\coloneqq\int_{G_d} d\nu(U) U^{\otimes t} \otimes \bar{U}^{\otimes t}\ .
\end{equation}
One can easily show that  $0\leq\delta(t,\nu_\mathcal{S}) \leq 1$ \cite{Nets}. When $\delta(t,\nu_\mathcal{S})=0$ we say that $\mathcal{S}$ is {\it an exact $t$-design} and when $\delta(t,\nu_\mathcal{S})=1$ we say that $\mathcal{S}$ is {\it not a $t$-design}.  In order to give an equivalent definition of an exact $t$-design \cite{BZ,gross}, recall the definition of an $S^t$-twirl of an operator $A$ is:
\begin{equation} \label{eq:altdef}
     \int_{G_d} d\nu_{\mathcal{S}}(U)
    U^{\otimes t} A \, (U^\dagger)^{\otimes t}.  
\end{equation}
 $\mathcal{S}$ is an exact $t$-design iff for any operator $A$  the $S^t$-twirl and the $G^t_d$-twirl coincides, i.e.,
\begin{equation} \label{eq:altdef}
     \int_{G_d} d\nu_{\mathcal{S}}(U)
    U^{\otimes t} A \, (U^\dagger)^{\otimes t}  =  \int_{G_d} \, d\mu(U) U^{\otimes t} A \, (U^\dagger)^{\otimes t} .
\end{equation}

To proceed, we note that a map $U\mapsto U^{\otimes t}\otimes \bar{U}^{\otimes t}$ is a representation of the unitary group $G_d$. This representation turns out to be reducible and as such it decomposes into some irreducible representations $\pi$ of $G_d$  
\begin{gather}\label{decomposition}
    U^{\otimes t}\otimes \bar{U}^{\otimes t}\simeq \bigoplus_{\pi}\pi(U)^{\oplus m_{\pi}}\simeq  \left(U\otimes \bar{U}\right)^{\otimes t},
\end{gather}
where $m_{\pi}$ is the multiplicity of $\pi$. The representations occurring in this decomposition are in fact irreducible representation of the projective unitary group, $PG_d=G_d/\sim$, where $U\sim V$ iff $U=e^{i\phi} V$. One can show that every irreducible representation of $PG_d$ arises this way for some, possibly large, $t$ \cite{Dieck}. For $t=1$ the decomposition (\ref{decomposition}) is particularly simple and reads $U\otimes \bar{U}=\mathrm{Ad}_U\oplus 1$, where $1$ stands for the trivial representation and $\mathrm{Ad}_U$ is the adjoint representation of $G_d$ and $PG_d\simeq\mathrm{Ad}_{G_d}$. 
%For a set of $\mathcal{S}\subset G_d$ we define
Next we define the notion of a group generated by a gate set.
\begin{definition}
Let $\mathcal{S}\subset G_d$ be a set of quantum gates. The group, $G_\mathcal{S}$, generated by $\mathcal{S}$ is  
\begin{gather}
    G_\mathcal{S}\coloneqq  \mathrm{cl} (<\mathcal{S}>),\label{closure}\\
    <\mathcal{S}>\coloneqq \bigcup_{l\in \mathbb{N}} \mathcal{S}_l,\,\,\ \mathcal{S}_l\coloneqq \{g_1g_2\ldots g_l|\, g_i\in \mathcal{S}\}. 
\end{gather}
\end{definition}
The closure in (\ref{closure}) corresponds to adding all the limiting points of the sequences of words from $\mathcal{S}_l$ when $l$ goes to infinity. 
\begin{definition}
A set $\mathcal{S}\subset G_d$ is universal if and only if $PG_d\simeq G_\mathcal{S}/\sim$, where $U\sim V$ iff $U=e^{i\phi} V$.
\end{definition}
Having any representation $\pi$ of $G_d$ we can consider its restriction to $G_\mathcal{S}$, which we denote by $\mathrm{Res}^{G_d}_{G_\mathcal{S}}\pi$. For any $U\in G_\mathcal{S}$ we simply have:
\begin{gather}
   \mathrm{Res}^{G_d}_{G_\mathcal{S}}\pi(U)=\pi(U).
\end{gather}
A crucial observation here is that $\mathrm{Res}^{G_d}_{G_\mathcal{S}}\pi$ can be a reducible representation of $G_\mathcal{S}$ even though $\pi$ is an irreducible representation of $G_d$. This observation plays a central role in representation theory (cf. branching rules) and turns out to be central in our context. For any irreducible representations occurring in the decomposition \eqref{decomposition} we can consider the restriction of $T_{\nu_{\mathcal{S}},t}$ to $\pi$ which is given by
\begin{gather}
    T_{\nu_{\mathcal{S}},t,\pi}=\int_{G_d} d\nu_\mathcal{S}(U)\pi(U).
\end{gather}
It follows directly from the  definitions and discussion above that 
\begin{gather}
    T_{\nu_{\mathcal{S}},t}\simeq \bigoplus_{\pi}\left (T_{\nu_{\mathcal{S}},t,\pi}\right)^{\oplus m_\pi},
    \\
    \delta(t,\nu_S)=\mathrm{sup}_{\pi}\|T_{\nu_{\mathcal{S}},t,\pi}\|_{\infty},
\end{gather}
where $\pi$ goes over irreducible representations occurring in the decomposition (\ref{decomposition}).

\begin{lemma}\label{remark1}
Let $\mathcal{S}\subset G_d$ be an arbitrary set of quantum gates. Then $\delta(t,\nu_{\mathcal{S}})=1$ if and only if for some nontrivial irreducible representation $\pi$ appearing in the decomposition (\ref{decomposition}) representation $\mathrm{Res}^{G_d}_{G_\mathcal{S}}\pi$ is reducible and contains a copy of the trivial representation. 
\end{lemma}
\begin{proof}
The implication `$\Leftarrow$' is obvious. For `$\Rightarrow$' assume that $\delta(t,\nu_{\mathcal{S}})=1$. Then for some irreducible nontrivial representation $\pi$ appearing in the decomposition \eqref{decomposition} we have $\|T_{\nu_{\mathcal{S}},t,\pi}\|=1$. There exists vector of norm one, $v\in V_\pi$, such that $\sum_{g,h\in\mathcal{S}}\pi(gh^\dagger)v=|\mathcal{S}|^2v$. Hence, $\sum_{g,h\in\mathcal{S}}\bk{v}{\pi(gh^\dagger)v}=|\mathcal{S}|^2$. By the unitarity of $\pi(gh^\dagger)$ we have  $|\bk{v}{\pi(gh^\dagger)v}|\leq 1$, for any $v\in V_\pi$ of norm one. Thus $\forall\,g,h\in \mathcal{S}\,\,\pi(g)v=\pi(h)v$. Note, however, that under the assumption that $I\in \mathcal{S}$, this implies $\forall\,g\in \mathcal{S}\,\,\pi(g)v=v$ which means $v$ is a common eigenvector of all operators $\pi(g)$, $g\in\mathcal{S}$. This in turn means that $\mathrm{Res}^{G_d}_{G_\mathcal{S}}\pi$ is reducible and contains a copy of the trivial representation. This ends the proof.
\end{proof}

%ALTERNATIVE PROOF BEGIN

%\begin{proof}
%The implication `$\Leftarrow$' is obvious. For `$\Rightarrow$' assume that $\delta(t,\nu_{\mathcal{S}}) = 1$. Then for some irreducible nontrivial representation $\pi$ appearing in the decomposition \eqref{decomposition} we have $\|T_{\nu_{\mathcal{S}},t,\pi}\| = 1$, i.e. $\|\sum_{g\in\mathcal{S}}\pi(g)\| = |\mathcal{S}|$.
%By the unitarity of $\pi(g)$ this is a saturated triangle inequality, which holds iff there exists normalized to one vector $v \in V_\pi$ such that $\pi(g) v$ is the same $\forall\,g\in \mathcal{S}$.
%Since $I \in \mathcal{S}$, it follows that $\forall\,g\in \mathcal{S}\,\,\pi(g)v=v$, which means $v$ is a common eigenvector of all operators $\pi(g)$, $g\in\mathcal{S}$. This in turn means that $\mathrm{Res}^{G_d}_{G_\mathcal{S}}\pi$ is reducible and contains a copy of the trivial representation. This ends the proof.
%\end{proof}

%ALTERNATIVE PROOF END

%\begin{remark}
%In the language of the alternative definition of exact $t$-designs given by Eq.~\eqref{eq:altdef}, the result of Lemma~1 means that an arbitrary set of quantum gates $\mathcal{S}\subset G_d$ forms an exact $t$-design if and only if  there exists an operator $A$ such that $A =\int_{G_d} d\nu_S(U)
%    U^{\otimes t} A \bar{U}^{\otimes t} \ne  \int_{G_d} d\mu(U) U^{\otimes t} A \bar{U}^{\otimes t}$.
%\end{remark}

\begin{remark}\label{remark1b}
The following is an equivalent way of formulating Lemma~1 using  Eq.~\eqref{eq:altdef}: $\{ \mathcal{S}, \nu_{\mathcal{S}} \}$ is not a \hbox{$t$-design} iff  there exists an operator $A$ such that $A =\int_{G_d} d\nu_{\mathcal{S}}(U)
    U^{\otimes t} A (U^\dagger)^{\otimes t} \ne  \int_{G_d} d\mu(U) U^{\otimes t} A (U^\dagger)^{\otimes t}$.
\end{remark}

\begin{Corollary}\label{gap}
Assume $\mathcal{S}$ is universal. Then  $\delta(t,\nu_{\mathcal{S}})<1$ for any finite $t$. 
\end{Corollary}
\begin{proof}
When $\mathcal{S}$ is a universal set, $PG_d=G_{\mathcal{S}}/\sim$ and hence the restriction of any irreducible representation $\pi$ of $PG_d$ to $G_\mathcal{S}$ remains irreducible.
\end{proof}

%We note that the converse of Corollary \ref{gap} is in general false as  $\delta(t,\nu_{\mathcal{S}})<1$  implies only that the restriction to $G_\mathcal{S}$ of any nontrivial irreducible representation of $PG_d$ does not contain the trivial representation in its decomposition. 
Let us next define 
\begin{align*}
    &\mathcal{S}^{t_1,t_2}=\{U^{\otimes t_1}\otimes \bar{U}^{\otimes t_2}|U\in \mathcal{S}\},\, \\
    &G_d^{t_1,t_2}=\{U^{\otimes t_1}\otimes \bar{U}^{\otimes t_2}|U\in G_d\}.
\end{align*}
We will denote $\mathcal{S}^{t,0} $ by $\mathcal{S}^{t} $ and $G_d^{t,0}$ by $G_d^{t}$.
For any set of matrices $B\subset \mathcal{B}(\mathbb{C}^n)$ let 
\begin{gather}\label{centralizer}
\mathcal{C}(B)=\{X\in \mathcal{B}(\mathbb{C}^n)|\,[X,Y]=0,\,\forall Y\in B\}
\end{gather}
Using $U\otimes \bar{U}=\mathrm{Ad}_U\oplus 1$ we can rewrite Lemma 3.4 from \cite{Adam} as
\begin{lemma}\label{lema1}
Assume $\mathcal{S}$ is such that 
\begin{gather}\label{necessary}
\mathcal{C}(\mathcal{S}^{1,1})=\mathcal{C}( G_d^{1,1}).
\end{gather}
Then $\mathcal{S}$ is universal if and only if $G_\mathcal{S}$ is infinite. 
\end{lemma}

In other words, whenever the condition  (\ref{necessary}) is satisfied, $\mathcal{S}$ is either universal or $G_\mathcal{S}$ is a finite subgroup of $G_d$. In what follows we will call the condition (\ref{necessary}) the necessary universality condition. We note also that we always have $\mathcal{C}(\mathcal{S}^{t,t})=\mathcal{C}(G_\mathcal{S}^{t,t})$. Lemma \ref{lema1} implies that finite groups satisfying the necessary universality condition play a central role in deciding whether $\mathcal{S}$ is universal. 

\begin{definition}
A finite subgroup $G\subset G_d$ is a unitary $t$-group iff $\delta(t,\nu_{G})=0$, where $\nu_G$ is the uniform measure on $G$.
\end{definition}

The concept of unitary $t$-groups will play a central role in what follows.

\begin{lemma}\label{t-group}
Assume that $\mathcal{S}$ is a generating set of a finite subgroup of $G_d$. Then for any $t$, either $\delta(t,\nu_{G_\mathcal{S}})=0$ or $\delta(t,\nu_{G_\mathcal{S}})=1$. Moreover, for any $t$ (1) either $\delta(t,\nu_\mathcal{S})<1$ iff  $\delta(t,\nu_{G_\mathcal{S}})=0$ or (2) $\delta(t,\nu_\mathcal{S})=1$ iff $\delta(t,\nu_{G_\mathcal{S}})=1$.
\end{lemma}

\begin{proof}
One can easily verify that for any measure $\nu$ $\delta(t,{\nu}^{\ast l})=\delta(t,\nu)^l$, where $\nu^{\ast l}$ is the $l$-fold convolution of measure $\nu$. Under the assumption that $G_\mathcal{S}$ is finite there is $l_0$ such that $\mathcal{S}_{l_{0}}=G_\mathcal{S}$. On the other hand, it is known \cite{Varju} that for $l\rightarrow \infty$ the measure $\nu_{\mathcal{S}}^{\ast l}$ converges to $\nu_{G_{\mathcal{S}}}$. Thus we have $\delta(t,\nu_\mathcal{S})^l\rightarrow \delta(t,\nu_{G_\mathcal{S}})$. Note, however, that $\nu_{G_\mathcal{S}}=\nu_{G_\mathcal{S}}^{\ast 2}$. Thus  $\delta(t,\nu_{G_\mathcal{S}})=0$ or $\delta(t,\nu_{G_\mathcal{S}})=1$. The result follows.
\end{proof}

Recently there has been some development in the theory of unitary $t$-groups. The main result of \cite{Bannai20} states: 

\begin{fact} \label{thmt-gr}
There are no  finite unitary $t$-groups in $G_d$ for $t\geq 6$ and  $d\geq 2$. Moreover 
\begin{enumerate}
    \item When $d=2$ there is unitary $5$-group but no unitary $t$-group with $t\geq 6$.
    \item When $d\geq 3$ there is no unitary $t$-group with $t\geq 4$.
\end{enumerate}
\end{fact}

This lead to our first main result.

\begin{theorem}\label{main}
Let $\mathcal{S}$ be a set of gates in $G_d$ such that $\mathcal{C}(\mathcal{S}^{1,1})=\mathcal{C}( G_d^{1,1})$. Then $\mathcal{S}$ is universal if and only if
\begin{enumerate}
    \item $\{\mathcal{S},\nu_{\mathcal{S}}\}$ is a $\delta$-approximate $6$-design with $\delta<1$, when $d=2$.
    \item $\{\mathcal{S},\nu_{\mathcal{S}}\}$ is a $\delta$-approximate $4$-design with $\delta<1$, when $d\geq 3$.
\end{enumerate}

\end{theorem}
\begin{proof}
By Lemma \ref{lema1} the set $\mathcal{S}$ can be either universal or $G_\mathcal{S}$ is a finite group. If $\mathcal{S}$ is universal then by Corollary \ref{gap} we have that $\delta(t,\nu_{\mathcal{S}})< 1$ for any $t\geq 1$. On the other hand if $G_\mathcal{S}$ is a finite group then by Lemma \ref{t-group} and  Fact \ref{thmt-gr} we can only have: (1) $\delta(t,\nu_{\mathcal{S}})=1$ for $t=6$ and $d=2$, (2) $\delta(t,\nu_{\mathcal{S}})=1$ for $t=4$ and $d>2$. This finishes the proof. 
\end{proof}
In order to state our second universality criterion we will need the following lemma.
\begin{lemma}\label{lema2}
Assume that for some $t\geq 2$
\begin{gather}\label{c1}
   \mathcal{C}(\mathcal{S}^{t,t})=\mathcal{C}( G_d^{t,t}).
\end{gather}
Then any irreducible representation occurring in the decomposition (\ref{decomposition}) remains irreducible when restricted to $G_\mathcal{S}$. Moreover  $\mathcal{C}(\mathcal{S}^{1,1})=\mathcal{C}( G_d^{1,1})$.
\end{lemma}
\begin{proof}
The fact that any irreducible representation occurring in the decomposition (\ref{decomposition}) remains irreducible when restricted to $G_\mathcal{S}$ is obvious. For the second part of the statement note that 
\begin{gather}
    U^{\otimes t}\otimes\bar{U}^{\otimes t}\simeq \left (U^{\otimes t-1}\otimes\bar{U}^{\otimes t-1}\right)\otimes \left(\mathrm{Ad}_U\oplus 1\right)=\\\nonumber=\left(\left (U^{\otimes t-1}\otimes\bar{U}^{\otimes t-1}\right)\otimes \mathrm{Ad}_U\right)\oplus \left (U^{\otimes t-1}\otimes\bar{U}^{\otimes t-1}\right).
\end{gather}
Repeating this decomposition we get that $U\otimes \bar{U}$ is one of the summands in the decomposition of $ U^{\otimes t}\otimes\bar{U}^{\otimes t}$. Hence, under condition (\ref{c1}) we have $\mathrm{Res}^{G_d}_{G_\mathcal{S}}\mathrm{Ad}$ is irreducible. Thus $\mathcal{C}(\mathcal{S}^{1,1})=\mathcal{C}( G_d^{1,1})$.
\end{proof}
We can now formulate a sufficient condition for universality in terms of centralizers.
\begin{Corollary}\label{main2}
Let $\mathcal{S}$ be a set of gates in $G_d$. Then $\mathcal{S}$ is universal iff for $t(2)=6$ and $t(d)=4$ for $d>2$ we have
\begin{gather}\label{criterion1}
\mathcal{C}(\mathcal{S}^{t(d),t(d)})=\mathcal{C}( G_d^{t(d),t(d)}).
\end{gather}
\end{Corollary}
\begin{proof}
The condition \eqref{criterion1} combined with Lemma \ref{lema2} implies that the necessary condition for universality is satisfied and that all irreducible representations occurring in the decomposition (\ref{decomposition}) remain irreducible when restricted to $G_\mathcal{S}$, where $t(d)$ is as in the statement of the theorem. Hence by Lemma \ref{remark1} we have $\delta(6,\nu_{\mathcal{S}})\neq 1$ for $d=2$ and $\delta(4,\nu_{\mathcal{S}})\neq 1$ for $d>2$. Assume $G_\mathcal{S}$ is a finite group. Then by  Lemma \ref{t-group}  and Fact \ref{thmt-gr} we have $\delta(6,\nu_{\mathcal{S}})=1$ for $d=2$ and $\delta(4,\nu_{\mathcal{S}})=1  $ for $d>2$. Thus we get a contradiction and $\mathcal{S}$ is universal. 
\end{proof}
Corollary~\ref{main2} can be further improved. For this, we will use a technique  connecting $\mathcal{C}(\mathcal{S}^{t_1,t_2})$ with $\mathcal{C}(\mathcal{S}^{t_1-n,t_2+n})$ through a partial transpose map. To be more concrete, we will use the following lemma  
\begin{lemma} \label{lem:par_transpose}
Let $\mathcal{S}$ be a set of gates in $G_d$, and let $\theta$ denote the transposition operator on $\mathcal{B}(\mathbb{C}^d)$. Then for any non-negative integer number $n \le t$, we have that 
\begin{gather}
{\rm{id}}^{\otimes (t{-}n)} \otimes \theta^{n}   (\mathcal{C}(\mathcal{S}^{t}))=\mathcal{C}(\mathcal{S}^{t-n,n}).
\end{gather}
In particular, the dimensions of $\mathcal{C}(\mathcal{S}^{t})$ and  $\mathcal{C}(\mathcal{S}^{t-n,n})$ are equal.
\end{lemma}
\begin{proof}
Let $X = \sum_i X^{1}_i \otimes X^{2}_i \cdots \otimes  X^{t}_i$  be an element of $ \mathcal{C}(\mathcal{S}^{t})$ (note that the upper index is indeed an index, not an exponent). Using the notations $m=t-n$ and $B^T=\theta(B)$ (and noting that $U^T=\bar{U}^\dagger$ for any unitary), we calculate the adjoint action of an arbitrary $U^{\otimes m} \otimes \bar{U}^{\otimes n} \in \mathcal{S}^{m,n}$ (with $U \in \mathcal{S}$) on ${\rm{id}}^{\otimes m} \otimes \theta^{n} (X) $:
\begin{align}
  &   U^{\otimes m} \otimes \bar{U}^{\otimes n} \, ({\rm{id}}^{\otimes m} \otimes \theta^{n} (X)) \,  (U^{\otimes m} \otimes \bar{U}^{\otimes n})^\dagger = \nonumber\\ 
  %  &    {\textstyle \sum_i} [U^{\otimes m} (X^{1}_i {\otimes} {\cdot} {\cdot} {\cdot} X^{m}_i )(U^\dagger)^{\otimes m}] {\otimes} (U^\dagger)^T (X^{m+1}_i)^T U^T {\otimes} {\cdot} {\cdot} {\cdot} \bar{U} (X_i^{t})^T  \bar{U}^\dagger {=} \nonumber \\ 
 % &    {\textstyle \sum_i} U X^{1}_i U^\dagger {\otimes} {\cdot} {\cdot} {\cdot} U X^{m}_i U^\dagger {\otimes} (U^\dagger)^T (X^{m+1}_i)^T U^T {\otimes} {\cdot} {\cdot} {\cdot} \bar{U} (X_i^{t})^T  \bar{U}^\dagger {=} \nonumber \\ 
  &    {\textstyle \sum_i} U X^{1}_i U^\dagger {\otimes} {\cdot} {\cdot} {\cdot} U X^{m}_i U^\dagger {\otimes} \bar{U} (X^{m+1}_i)^T \bar{U}^\dagger {\otimes} {\cdot} {\cdot} {\cdot} \bar{U} (X_i^{t})^T  \bar{U}^\dagger {=} \nonumber \\ 
  &      {\textstyle \sum_i} U X^{1}_i U^\dagger {\otimes} {\cdot} {\cdot} {\cdot} {\otimes}  (U^\dagger)^T (X^{m+1}_i)^T U^T {\otimes} {\cdot} {\cdot} {\cdot}   (U^\dagger)^T (X_i^{t})^T  U^T {=} \nonumber \\ 
   &            {\textstyle \sum_i} U X^{1}_i U^\dagger {\otimes} {\cdot} {\cdot} {\cdot} U X^{m}_i U^\dagger {\otimes}  (U X^{m+1}_i U^\dagger)^T {\otimes} {\cdot} {\cdot} {\cdot}   (U X_i^{t} U^\dagger)^T {=} \nonumber \\ 
  &  {\textstyle \sum_i}  {\rm{id}}^{\otimes m} \otimes \theta^{n} ( U X^{1}_i U^\dagger \otimes \cdots \otimes  U X^{t}_i U^\dagger) = \nonumber\\
  &   {\rm{id}}^{\otimes m} \otimes \theta^{n} (U^{\otimes t} X (U^\dagger)^{\otimes t}) = {\rm{id}}^{\otimes m} \otimes \theta^{n} ( X ),
\end{align}
where the last equality followed from the fact that $X$ commutes with $U^{\otimes t}$. This means that  for any $X \in \mathcal{C}(\mathcal{S}^{t})$ we have  ${\rm{id}}^{\otimes m} \otimes \theta^{n} (X) \in  \mathcal{C}(\mathcal{S}^{t-n, n})$. To prove that all the elements of $ \mathcal{C}(\mathcal{S}^{t-n, n})$ can be obtained this way, one can note that $ ({\rm{id}}^{\otimes m} \otimes \theta^{n}) \circ ({\rm{id}}^{\otimes m} \otimes \theta^{n}) = {\rm{id}}^{\otimes t}$ and then repeat a completely analogous proof for showing that $Y \in \mathcal{C}(\mathcal{S}^{t-n, n})$ implies $({\rm{id}}^{\otimes m} \otimes \theta^{n}) (Y) \in \mathcal{C}(\mathcal{S}^{t})$.
\end{proof}

\begin{theorem}\label{main3}
Let $\mathcal{S}$ be a set of gates in $G_d$. Then $\mathcal{S}$ is universal if and only if:
\begin{gather}\label{criterion2}
\mathcal{C}(\mathcal{S}^{t(d),t(d)})=\mathcal{C}( G_d^{t(d),t(d)}),
\end{gather}
where $t(2)=3$, and $t(d)=2$ for $d\geq 3$. 
\end{theorem}
\begin{proof}
Suppose that the set of gates $\mathcal{S} \subset G_{d}$ is non-universal, and denote by $G_{\mathcal{S}}$ the group generated by ${\mathcal{S}}$. If $G_{\mathcal{S}}$ is infinite, then Lemma 2 guarantees that $\mathcal{C}(\mathcal{S}^{1,1}) \ne\mathcal{C}( G_d^{1,1})$, hence also $\mathcal{C}(\mathcal{S}^{n,n}) \ne\mathcal{C}( G_d^{n,n})$ for any positive integer $n$.  If $G_{\mathcal{S}}$ is finite, then we know that it cannot form a $k(d)$-design with $k(2)=6$ and $k(d)=4$ for $d \ge 3$. From  Remark~\ref{remark1b} it follows that there exists an operator $A$ such that  $A =\int_{G_d} d\nu_{G_\mathcal{S}}(U)
    U^{\otimes k(d)} A (U^\dagger)^{\otimes k(d)} \ne  \int_{G_d} d\mu(U) U^{\otimes k(d)} A (U^\dagger)^{\otimes k(d)}$. On the one hand, $A \in \mathcal{C}( \mathcal{S}^{k(d)})$,
    since every $G^{k}_{\mathcal{S}}$-twirled element commutes with the elements of  $G_{\mathcal{S}}^{k}$ \cite{diaconis}. On the other hand, $A \notin \mathcal{C}( G_d^{k(d)})$ since A is not equal to its $G^{k(d)}_d$-twirl. Thus,  $\mathcal{C}( \mathcal{S}^{k(d)}) \ne \mathcal{C}( G_d^{k(d)})$. Now using the last sentence of Lemma~\ref{lem:par_transpose}, we get that
    $\mathcal{C}( \mathcal{S}^{t(d),t(d)}) \ne \mathcal{C}( G_d^{t(d),t(d)})$, where $t(d)=k(d)/2$. Therefore, if $\mathcal{C}(\mathcal{S}^{t(d),t(d)}) = \mathcal{C}( G_d^{t(d),t(d)})$, then $\mathcal{S}$ has to be a universal gate set. This concludes the proof.
\end{proof}

One can calculate explicitly the dimension of the centralizer $\mathcal{C}( G_d^{t,t})$ from the following formula \cite{Stroomer94}
\begin{gather}\label{dim}
\dim \mathcal{C}( G_d^{t,t})=\sum_{\pi}(m_\pi)^2,
\end{gather}
where $\pi$ are irreducible representations occurring in the decomposition (\ref{decomposition}). Following Corollary 5.4 of \cite{Stroomer94} we know that 
\begin{gather}\label{d8}
\dim \mathcal{C}( G_d^{t,t})=(2t)!\, ,\,\, d\geq 2t.
\end{gather}
We can reformulate Theorem \ref{main3} to a more computationally friendly form:
\begin{theorem}\label{main4}
Let $\mathcal{S}$ be a set of gates in $G_d$. Then $\mathcal{S}$ is universal if and only if: (1) $\dim \mathcal{C}( {S}^{3,3})=132$ for $d=2$; (2) $\dim \mathcal{C}( \mathcal{S}^{2,2})=23$ for $d=3$; (3) $\dim \mathcal{C}( \mathcal{S}^{2,2})=4!=24$ for $d\geq 4$.
\end{theorem}
\begin{proof}
Obviously, for any $\mathcal{S}\subset G_d$ we have $\dim \mathcal{C}( {S}^{t,t})\geq \dim \mathcal{C}( G_d^{t,t})$. The equality of these dimensions is possible if and only if the restrictions to $G_\mathcal{S}$ of all irreducible representations occurring in the decomposition (\ref{decomposition}) are irreducible. Thus $\dim \mathcal{C}( {S}^{t,t})= \dim \mathcal{C}( G_d^{t,t})$ if and only if $ \mathcal{C}( {S}^{t,t})= \mathcal{C}( G_d^{t,t})$. We are left with finding dimensions of $\mathcal{C}( G_d^{t(d),t(d)})$, where $t(d)$ is as in the statement of Theorem \ref{main3}. For $d\geq 4$ we use identity (\ref{d8}). For $d=2$ we have \begin{gather}\label{decompositionSU2}
    U^{\otimes t}\otimes \bar{U}^{\otimes t}=\bigoplus_{0\leq\nu\leq{2t},\,\nu - \mathrm{even}} m_{\pi_{\nu}}\pi_{\nu}(U), 
\end{gather}
where $\mathrm{dim\pi_{\nu}}=\nu+1$ and $m_{\pi_{\nu}}\neq 0$ for every even $\nu$ satisfying $0\leq\nu\leq 2t$. Moreover, $\mathrm{Ad}_{U}=\pi_2(U)$ and 
\begin{gather}\label{decompositionSU22}
    U^{\otimes t}\otimes \bar{U}^{\otimes t}\simeq (\pi_2(U)\oplus 1)^{\otimes t}. 
\end{gather}
In addition we known that $\pi_{l}(U)\otimes\pi_k(U)\simeq \pi_{l+k}\oplus \pi_{l+k-2}\oplus\ldots\oplus\pi_{|l-k|}$. Using these identities we find that
\begin{gather}\label{decompositionSU21}
    U^{\otimes 3}\otimes \bar{U}^{\otimes 3}\simeq \pi_{6}(U)\oplus 5\pi_4(U)+9\pi_{2}(U)\oplus 5. 
\end{gather}
Thus using formula (\ref{dim}) we get $\dim \mathcal{C}( G_2^{3,3})=132$. A decomposition similar to (\ref{decompositionSU2}) can be found for $PG_3$ \cite{Stroomer94,Roy09}. We know that irreducible representations of $PG_3$ are indexed by pairs of non-negative integers $\lambda=(\lambda_1,\lambda_2)$, where $\lambda_1\geq\lambda_2\geq0$. The adjoint representation of $G_3$ has $\lambda = (2,1)$. Using the rules for decomposing a tensor product of $\pi_{(\lambda_1,\lambda_2)}\otimes \pi_{(\lambda_1^\prime,\lambda_2^\prime)}$ into irreducible representations \cite{Stroomer94,Roy09} we found that for $d=3$
\begin{gather}\label{decompositionSU3}
    U^{\otimes 2}\otimes \bar{U}^{\otimes 2}\simeq \pi_{(4,2)}\oplus\pi_{(3,0)}\oplus \pi_{(3,3)}\oplus 4\pi_{(2,1)}\oplus 2,
\end{gather}
Thus using (\ref{dim}) we get
$\dim \mathcal{C}( G_3^{2,2})=23$. 
%where $\dim \pi_{(4,2)} = 27$, $\dim \pi_{(3,0)} = \dim \pi_{(3,3)} = 10$ and $\dim \pi_{(2,1)} = 8$.
\end{proof}

%In conclusion, our new criterion allows to 

Finally, we note that calculating $\dim \mathcal{C}( \mathcal{S}^{t(d),t(d)})$ boils down to solving a set of $|\mathcal{S}|d^{2t(d)}$ linear equations, where $t(d)$ is as in Theorem \ref{main3}.  Thus, the number of linear equations to be solved, for determining whether a gate set $\mathcal{S} \subset U(d)$ is universal or not, scales at most as  $O(d^{4})$ for $d \ge 3$. This means that our direct universality check is much more efficient and more feasible than the previously developed methods \cite{Adam, Sawicki2, Nets} with $O(d^{d^{5/2}})$ scaling. One can also compare our results with \cite{babai} that gives an algorithm that decides finitness of a matrix group in the polynomial time in $d$. This approach, however, is based on more mathematically abstract concepts, such as algebraic field extensions. Our method, on the other hand, is direct and uses concepts that are meaningful from the perspective of quantum circuits design, i.e. $t$-designs. For example, a gate-set in dimension $d\geq 3$ that is $\delta_4$-approximate $4$-design with $\delta_4$ very close to one is also very inefficient in terms of approximating unitaries. This follows from the fact that such a gate-set is also $\delta_t$-approximate $t$-design with $\delta_t \geq \delta_4$, for any $t>4$, and the  depth of $\epsilon$-approximating circuit is known to be inversely proportional to $1-\delta_t$, where $t=O(\epsilon^{-1})$ (see \cite{slowik-sawicki} for more details). Thus from the point of view of applications, gate-sets that are $\delta$-approximate $4$-designs with $\delta$ very close to one can be regarded as nonuniversal. Moreover, we also expect that our simple algebraic criterion for universality will allow general proofs about universal extension of different gate-set families going well beyond the earlier results \cite{Brylinski, Oszmaniec1,Mick}. We leave this as future work.\\

%as future work the utilization of our criterion to determine the universal extensions of relevant gate sets. \\

\noindent {\it Acknowledgements.} AS would like to thank M. Oszmaniec for ingenious discussions regarding Lemma~3. This work was supported by National Science Centre, Poland under the grant SONATA BIS:2015/18/E/ST1/00200, and by the Hungarian National Research, Development and Innovation Office through the Quantum Information National Laboratory program and Grants No. FK135220, KH129601, K124351.

\bibliographystyle{apsrev4-2}

\begin{thebibliography}{6}
\bibitem{Barenco} A. Barenco {\it et al.}, Phys. Rev. A \textbf{52}, 3457 (1995).
\bibitem{Dojcz} D. Deutsch, Proc. Roy. Soc. Lond. A \textbf{425}, 73--90 (1989).
\bibitem{Loyd} S. Lloyd, Phys. Rev. Lett. \textbf{75}, 346 (1995).
\bibitem{Reck} M. Reck, A. Zeilinger, H.~J. Bernstein, P. Bertani, Phys. Rev. Lett. \textbf{73}, 58 (1994).
\bibitem{exp1} Y. Bromberg, Y. Lahini, R. Morandotti, Y. Silberberg , Phys. Rev. Lett. \textbf{102}, 253904 (2009).
\bibitem{exp2} A. Politi {\it et al.}, Science \textbf{320},  646--649 (2008).
\bibitem{harrow} A.~W.~Harrow, B.~Recht, I.~L.~Chuang, J. Math. Phys. \textbf{43}, 4445--4451 (2002). 
\bibitem{Sawicki1} A. Sawicki, Quantum Inf. Comput. \textbf{16}, 291--312 (2016).
\bibitem{Sawicki3} K. Karnas, A. Sawicki, J. Phys. A: Math. Theor. \textbf{51}, 075305 (2018).
\bibitem{BA} A. Bouland, S. Aaronson, Phys. Rev. A \textbf{89}, 062316 (2014).
\bibitem{selinger} P. Selinger, Quant. Inf. Comp. \textbf{15}, 159--180 (2015).
\bibitem{klucz} V. Kliuchnikov, D. Maslov, M. Mosca, IEEE Transactions on Computers \textbf{65}, 161--172 (2016).
\bibitem{bocharov} A. Bocharov, Y. Gurevich, K. M. Svore, Phys. Rev. A \textbf{88}, 012313 (2013).
\bibitem{sarnak} P. Sarnak, Letter to Scott Aaronson and Andy Pollington on the Solovay-Kitaev theorem, (2015)
\bibitem{Brylinski} J.~L.~Brylinski,  R.~Brylinski, in {\it Mathematics of quantum computation}, pp. 117--134 (Chapman and Hall/CRC Pres, 2011).
\bibitem{NC00} M. Nielsen, I. Chuang, {\it Quantum Computation and Quantum Information} (Cambridge University Press, 2000).
\bibitem{Oszmaniec1} M. Oszmaniec, Z. Zimbor\'as, Phys. Rev. Lett. \textbf{119}, 220502 (2017).
\bibitem{daniel} Z. Zimbor\'as,  R. Zeier, T. Schulte-Herbr\"uggen, D. Burgarth, Phys. Rev. A \textbf{92}, 042309 (2015).
\bibitem{zeier1} R. Zeier, T. Schulte-Herbr\"uggen, J. Math. Phys. \textbf{52}, 113510 (2011).
\bibitem{zeier} R. Zeier, Z. Zimbor\'as, J. Math. Phys. \textbf{56}, 081702 (2015).
\bibitem{Laura} A. M Childs, D. Leung, L. Man\v{c}inska, M. Ozols, Quantum Inf. Comput. \textbf{11}, 19--39 (2011).
\bibitem{Claudio} C. Altafini, J. Math. Phys. \textbf{43}, 2051--2062 (2002).
\bibitem{Dallesandro}F. Albertini and D. D'Alessandro, IEEE Transactions on Automatic Control \textbf{48}, 1399--1403 (2003).
\bibitem{Adam} A. Sawicki, K. Karnas, Annales Henri Poincar{\'e} \textbf{18}, 3515--3552 (2017).
\bibitem{Sawicki2} A. Sawicki, K. Karnas, Phys. Rev. A \textbf{95}, 062303 (2017).
\bibitem{Mattioli-Sawicki} L. Mattioli, A. Sawicki, arXiv:2110.04210 (2021).
\bibitem{Freedman} M. H. Freedman, A. Kitaev, J. Lurie, Math. Res. Lett. \textbf{10}, 11--20 (2003).
\bibitem{Nets} M.~Oszmaniec, A.~Sawicki, M.~Horodecki, in IEEE Transactions on Information Theory, doi: 10.1109/TIT.2021.3128110 (2021).
\bibitem{Bannai20} E. Bannai,
G. Navarro, N. Rizo, P.~H.~Tiep, J. Math. Soc. Japan \textbf{72}, 909--921 (2020).
\bibitem{BZ}I. Bengtsson, K. \.Zyczkowski, {\it Geometry of quantum states: an introduction to quantum entanglement} (Cambridge University Press, 2017).
\bibitem{gross}D. Gross, K. Audenaert, J. Eisert, J. Math. Phys., \textbf{48}, 052104 (2007).
\bibitem{Dieck} T. Br\"ocker, T. Dieck, {\it Representations of Compact Lie Groups}, (Springer-Verlag, New York, 1985).
\bibitem{Varju}P. P. Varj\'u, Doc. Math. \textbf{18}, 1137--1175 (2013).
\bibitem{diaconis} P. Diaconis, {\it Group Representations in Probability and Statistics}, Lecture Notes-Monograph Series Vol. 11 (Institute of Mathematical Statistics, 1988).
\bibitem{Stroomer94} 
G. Benkart, M. Chakrabarti, T. Halverson, R. Leduc, C.~Y. Lee, J. Stroomer, J. Algebra \textbf{166}, 529--567 (1994).
\bibitem{Roy09} A. Roy, A. J. Scott, Des. Codes Cryptogr. \textbf{53}, 13-31 (2009).
\bibitem{babai} L. Babai, R. Beals, D. N. Rockmore, Proc. of International Symposium on Symbolic and Algebraic Computation. ISSAC 93. ACM Press, pp. 117-126, (1993).
\bibitem{slowik-sawicki} O. S\l{}owik, A. Sawicki, 	arXiv:2201.11774, (2021).
\bibitem{Mick} M.~J.~Bremner, J.~L.~Dodd, M.~A.~Nielsen,  D.~Bacon, Phys. Rev. A, \textbf{69}, 012313 (2004).
\end{thebibliography}

\end{document}